\RequirePackage{amsmath}
\documentclass[11pt,envcountsame,a4paper]{llncs}

\usepackage[UKenglish]{babel}
\usepackage[utf8]{inputenc}
\usepackage{ltlback}

\author{Arne Meier\inst{1}\fnmsep\thanks{Supported by DFG grant ME 4279/1-1.}\and Sebastian Ordyniak\inst{2}\and M. S. Ramanujan\inst{2}\and Irena Schindler\inst{1}\fnmsep$^\star$}
\institute{Leibniz Universität Hannover, \email{$\{$meier.schindler$\}$@thi.uni-hannover.de} \and TU Wien, \email{sordyniak@gmail.com,Ramanujan.Sridharan@ii.uib.no}}
\title{Strong Backdoors for Linear Temporal Logic}

\longtrue

\begin{document}
\maketitle

\begin{abstract}
	In the present paper we introduce the notion of strong backdoors into the field of temporal logic for the CNF-fragment of linear temporal logic introduced by Fisher. We study the parameterised complexity of the satisfiability problem parameterised by the size of the backdoor. We distinguish between backdoor detection and evaluation of backdoors into the fragments of horn and krom formulas. Here we classify the operator fragments of globally-operators for past or future, and the combination of both. Detection is shown to be in FPT whereas the complexity of evaluation behaves different. We show that for krom formulas the problem is paraNP-complete. For horn formulas the complexity is shown to be either fixed parameter tractable or paraNP-complete depending on the considered operator fragment.
\end{abstract}

\section{Introduction}
Temporal logic is one of the most important formalisms in the area of program verification and validation of specification consistency. Most notably are the seminal contributions of Kripke \cite{kri63}, Pnueli \cite{pn77}, Emerson, Clarke, and Halpern \cite{emha85,clem81} to name only a few. There exist several different variants of temporal logic from which the best known ones are the computation tree logic CTL, the linear temporal logic LTL, and the full branching time logic CTL$^*$. In this paper we will consider a specific clausal fragment of LTL which is known as \emph{separated normal form} (SNF) and has been introduced by Fisher \cite{fish91}. This normal form is a generalisation of the conjunctive normal form from propositional logic to linear temporal logic with future and past modalities interpreted over the frame of the integers $(\mathbb Z,<)$. In SNF the formulas are divided into a past, present, and a future part. Technically this normal form is no restriction as always one can translate an arbitrary LTL formula to a satisfiability equivalent formula in SNF \cite{fish91}.

Sistla and Clarke have shown that the satisfiability problem for the logic LTL with its standard operators (\emph{next-time}, \emph{eventually}, \emph{always}, and \emph{until}) is $\PSPACE$-complete \cite{sc82}. Several restrictions of this generic problem have been considered so far: Ono and Nakamura \cite{on80} investigated operator fragments, Chen and Lin \cite{cl93} classified Horn formulas, Demri and Schnoebelen \cite{ds02} pinned the complexity down to tree parameters (temporal operator fragments, temporal depth, and number of propositional variables), Markey \cite{markey04} analysed the use of negation, Dixon~et al.\ \cite{dfk07} introduced an XOR fragment, Bauland et~al.\ \cite{bsssv09} applied the framework of Post's lattice together with operator fragments, and Artale et~al.\ \cite{ArtaleKRZ13} studied the SNF fragment in detail.

Whenever problems become classified for an intractable complexity class as for instance $\NP$ then there are different approaches to overcome this tractability defect. One prominent approach is the framework of parameterised complexity \cite{df13,dofe99}. Here the approach encompasses the identification of a \emph{parameter} of the instance such that if the value of the parameter is assumed to be small then the problem becomes tractable. Such a much desired parameter for propositional logic is the number of variables of a given formula $\varphi$. It is well-known that the (unparameterised) satisfiability problem of this logic is $\NP$-complete \cite{coo71a,levin73} whereas the mentioned parameterised point of view leads to a deterministic algorithm running in time $O(2^{|\Vars{\varphi}|}\cdot|\varphi|)$ and hence is, if $|\Vars{\varphi}|$ is fixed, a polynomial running time. This kind of algorithms leads to a class of problems which is said to be \emph{fixed parameter tractable} and the corresponding complexity class is called $\FPT$. However there is also a notion of intractability in the parameterised world, i.e., the complexity class $\W1$ or more general the $\W{}$-hierarchy. This hierarchy is located above of $\FPT$ and not known to be different. Similarly there is a nondeterministic variant of the class $\FPT$, namely, $\paraNP$ which contains the whole $\W{}$-hierarchy.

Until today many different types of parameterisations for SAT have been considered~\cite{sz03,ccfxjkg04,ops13}. One of these approaches making SAT fixed parameter tractable is the concept of backdoors \cite{GaspersSzeider12}. Informally a backdoor into some class of formulas $\mathcal C$ is a set of variables $X$ from a given formula $\varphi$ such that every assignment with respect to the variables in $X$ applied to $\varphi$ leads to a formula from the class $\mathcal C$. If $\mathcal C$ is then a tractable class of formulas like for instance $\HORN$ or $\KROM$ formulas then again we found such a fruitful parameterisation: the size of $X$. There are different variants of backdoors beyond propositional logic \cite{GaspersSzeider12} which have been considered, e.g., backdoors with respect to answer set programming \cite{fs15} or disjunctive logic programs \cite{fs15b}. It is worth to note that standard parameterisations, e.g., temporal depth, number of propositional variables and formula treewidth, do not help for $\LTL$ \cite{lm15}.

In this paper we introduce a notion of backdoors for $\LTL$ formulas and investigate the parameterised complexity of two problems: backdoor detection and evaluation. The backdoor classes of our interest are $\HORN$ and $\KROM$. Our classification shows that backdoor detection turns out to be in both cases in $\FPT$ whereas backdoor evaluation is a more challenging task: either it is $\FPT$ or $\paraNP$-hard depending on the chosen target formula class and considered operator fragment.

\iflong 
\else
 Some proofs have been omitted due to space restrictions and can be found in the technical report \cite{mso16}.
\fi

\section{Preliminaries}

\paragraph{Parameterised Complexity.} A \emph{parameterised problem}
$\Pi$ is a tuple $(Q,\kappa)$ such that the following
holds. $Q\subseteq\Sigma^*$ is a language over an alphabet $\Sigma$,
and $\kappa\colon\Sigma^*\to\N$ is a computable function; then
$\kappa$ also is called the \emph{parameterisation (of $\Pi$)}. 

If there is a deterministic Turing machine $M$ and a computable
function $f\colon \N \to \N$ s.t.~for every instance $x \in \Sigma^*$
(i) $M$ decides correctly if $x \in Q$, and (ii) $M$ has a runtime
bounded by $f(\kappa(x)) \cdot |x|^{O(1)}$, then we say that $M$ is an
\emph{fpt-algorithm for $\Pi$} and that $\Pi$ is \emph{fixed-parameter
  tractable} (or in the class $\FPT$).

%
%

The class $\paraNP$ contains all parameterised problems $(Q,\kappa)$ for which there is a computable function $f$ and an non-deterministic Turing machine deciding if $x \in Q$ holds in time $f(\kappa(x)) \cdot |x|^{O(1)}$. One way to show $\paraNP$-hardness of some parameterised problem $(Q,\kappa)$ is to show that $Q$ is $\NP$-hard for one specific, fixed value of $\kappa$, i.e., there exists a constant $\ell\in\N$ such that $(Q, \kappa)_\ell \dfn \{ x \mid x \in Q \text{ and } \kappa(x) = \ell \}$ is $\NP$-hard.

\paragraph{Temporal Logic.} We directly define the relevant fragment of well-formed formulas of linear temporal logic $\LTL$ in separated normal form (SNF) and stick to the notion of Artale et~al.~\cite{ArtaleKRZ13}.
\begin{align}
	\lambda &\ddfn \bot \mid p \mid \BoxF\lambda \mid \BoxP\lambda \mid \Boxstar\lambda,\\
	\varphi &\ddfn \lambda \mid \lnot\lambda \mid \varphi\land\varphi \mid \Boxstar(\lnot\lambda_1\lor\cdots\lor\lnot\lambda_n\lor\lambda_{n+1}\lor\cdots\lambda_{n+m}),
\end{align}
where $\lambda$ is also called \emph{temporal literal} and $\varphi$ is in \emph{clausal normal form}. The operators $\BoxF$, $\BoxP$, and $\Boxstar$ are read \emph{always in the future}, \emph{always in the past}, and \emph{always}.

We interpret the $\LTL$ formulas over the flow of time $(\Z,<)$ \cite{ghr94}. 
\begin{definition}[temporal semantics]
  Let $\PROP$ be a countable infinite set of propositions. A
  \emph{temporal interpretation} $\M=(\Z,<,V)$ is a mapping from
  propositions to moments of time, i.e.,
  $V\colon\PROP\to\Z$. The satisfaction relation $\models$ is
  then defined as follows where $n\in\Z$, $\varphi,\psi\in\LTL$

  \begin{tabbing}
    $\M,n\models \varphi\lor\psi$ \= iff \= $\M,n\models\varphi$ or $\M,n\models\psi$\kill
    $\M,n\models p$ \> iff \> $n\in V(p)$,\\
    $\M,n\models \varphi\lor\psi$ \> iff \> $\M,n\models\varphi$ or $\M,n\models\psi$\\
    $\M,n\models \lnot\varphi$ \> iff \> $\M,n\not\models\varphi$\\
    $\M,n\models \BoxF\varphi$ \> iff \> for all $k>n$ it holds $\M,k\models\varphi$\\
    $\M,n\models \BoxP\varphi$ \> iff \> for all $k<n$ it holds $\M,k\models\varphi$\\
    $\M,n\models \Boxstar\varphi$ \> iff \> for all $k\in\Z$ it holds $\M,k\models\varphi$
  \end{tabbing}
  
  We say that $\varphi$ is \emph{satisfiable} if there exists a
  temporal interpretation $\M$ such that
  $\M,0\models\varphi$. Then $\M$ is also referred to as a
  \emph{(temporal) model (of $\varphi)$}.
\end{definition}

Sometimes we also directly write $\M(p)$ instead of explicitly defining $V$. Note that the operator name $\mathsf{G}$ instead of $\BoxF$ often occurs in literature. 
We distinguish fragments of $\LTL$ by adding superscripts and subscripts as follows. If $O\subseteq\{\BoxF,\BoxP,\Boxstar\}$ is an operator subset then $\LTL^O$ is the fragment of $\LTL$ consisting only of formulas that are allowed to use temporal operators from $O$. We also consider restrictions on the clausal normal form for 
$
\Boxstar(\lnot\lambda_1\lor\cdots\lor\lnot\lambda_n\lor\lambda_{n+1}\lor\cdots\lambda_{n+m})
$ in (2). Table~\ref{tbl:forms} lists the relevant cases for this study. If $\alpha\in\{\CNF,\HORN,\KROM\}$ then $\LTL_\alpha$ is the set of formulas using the normal form $\alpha$.

\begin{table}[t]
	$$\begin{array}{lp{7cm}l}\toprule
		\text{class} &\text{description}& \text{restrictions on }n,m\\\midrule
		$\CNF$ &no restrictions on (2)& -\\
		$\HORN$&at most one positive temporal literal& m\leq 1\\
		$\KROM$&binary clauses& n+m\leq 2\\
	\bottomrule
	\end{array}$$
	
	\caption{Considered normal forms. Restrictions refer to equation (2).}\label{tbl:forms}
\end{table}

As shown by Fisher et~al.~every $\LTL$ formula considered over the frame $(\mathbb Z,<)$ has an satisfiability-equivalent formula in SNF \cite{fdp01}. 

\begin{lemma}[{\cite[Lemma 2]{ArtaleKRZ13}}]\label{lem:normalform}
	Let $\mathcal L\in\{\LTL_\alpha^{\{\BoxF,\BoxP\}},\LTL_\alpha^{\supBoxstar}\}$ for $\alpha\in\{\CNF,\HORN,\KROM\}$. For any formula $\varphi\in\mathcal L$, one can construct, in log-space, an satisfiability-equivalent $\mathcal L$-formula $\Psi\land\Boxstar\Phi$, where $\Psi$ is a conjunction of propositional variables from $\Phi$, and $\Phi$ is a conjunction of clauses of the form (2) containing only $\BoxF,\BoxP$ for $\LTL_\alpha^{\BoxF,\BoxP}$, and only $\Boxstar$ for $\LTL_\alpha^{\supBoxstar}$, in which the temporal operators are not nested.
\end{lemma}

In the following we consider all formulas given in that normal form. Therefore we assume that all formulas in $\LTL^O_\alpha$ obey this normal form.

\section{Introduction of strong backdoors for LTL}\label{sec:backdoor}

In the following we will introduce a notion of backdoors for formulas in linear temporal logic. 
The definition of these backdoors turns out to be very similar to the definition of the corresponding backdoor sets for propositional formulas. The main difference is that whenever a propositional variable is in the backdoor set then also all of its temporal literals is required to be in the backdoor set as well. A consequence of this is that in contrast to propositional formulas, where a backdoor set needs to consider all assignments of the backdoor set variables, we only need to consider assignments that are consistent between propositional variables and their temporal literals.

Let $\OOO$ be a set of operators. An assignment $\theta\colon \Vars\phi \cup \SB x \in \Vars\phi \land Ox \SM O \in \OOO\SE$ is \emph{consistent} if for every $x\in\Vars\phi$ it holds that if $\theta(\Boxstar x)=1$, then also $\theta(\Box_P x)=1$, $\theta(\Box_F x)=1$, and $\theta(x)=1$.

\begin{definition}[Backdoors]
  Let $\mathcal C$ be a class of $\CNF$-formulas, $\mathcal O$ be a set of operators, and $\phi$ be a $\LTL^{\mathcal O}_\CNF$ formula. 
  A set $X\subseteq\Vars{\phi}$ is a \emph{(strong) $\mathcal C$-backdoor} if for every consistent assignment $\theta\colon X\cup\{Ox\mid x\in X, O\in\mathcal O\}\to\{0,1\}$ it holds that $\phi[\theta]$ is in $\mathcal C$.
\end{definition}

The \emph{reduct} $\phi[\theta]$ is defined similar as for usual $\CNF$-formulas, i.e., all clauses which contain a satisfied literal are deleted, and all falsified literals are deleted from their clauses. Here empty clauses are substituted by false, and the empty formula by true. Sometimes if the context of $\mathcal O$ is clear we omit stating it and just mention the backdoor class $\mathcal C$.

In order to exploit backdoor sets to obtain efficient (FPT) algorithms for LTL one needs
to accomplish two tasks: First one needs to find a small backdoor set and then one needs
to show how the backdoor set can be exploited to efficiently evaluate the formula. 
This leads to the following problem definitions for
every class $\mathcal C$ of formulas and set of operators
$\mathcal O$.
\begin{description}
	\item[Problem:] $\BDeval{\mathcal O}{\mathcal C}$ --- Backdoor evaluation to $\LTL^{\mathcal O}_{\mathcal{C}}$.
	\item[Input:] $\LTL^{\mathcal O}_{\CNF}$ formula $\phi$, strong $(\mathcal C, \mathcal O)$-backdoor $X$.
	\item[Parameter:] $|X|$.
	\item[Question:] Is $\phi$ satisfiable?
\end{description}

\begin{description}
	\item[Problem:] $\Detect{\mathcal O}{\mathcal C}$ --- Backdoor detection to $\LTL^{\mathcal O}_{\mathcal{C}}$.
	\item[Input:] $\LTL^{\mathcal O}_{\CNF}$ formula $\phi$, integer $k\in\N$.
	\item[Parameter:] $k$.
	\item[Task:] Find a strong $\mathcal C$-backdoor of size $\leq k$ if one exists.
\end{description}

Of course this approach is only meaningful if one considers target classes which have polynomial time solvable satisfiability problems. Artale et~al.\ have shown \cite{ArtaleKRZ13} that satisfiability for $\LTLstar_{\HORN}$ and $\LTLstar_\KROM$ are solvable in $\P$. Adding $\BoxF,\BoxP$ to the set of allowed operators makes the $\KROM$ fragment $\NP$-complete whereas for $\HORN$ formulas the problem stays in $\P$. Therefore we will consider in the following only $\KROM$ and $\HORN$ formulas.

\section{Backdoor set detection}

In this section we show that finding strong $\CCC$\hy backdoor sets is fixed-parameter tractable if $\CCC$ is either $\HORN$, or $\KROM$. The algorithms that we will present are very similar to the algorithms that are known for the detection of strong backdoors for propositional CNF formulas~\cite{GaspersSzeider12}.

We first show how to deal with the fact that we only need to consider consistent assignments. The following observation is easily witness by the fact that if one of $\BoxP x,\BoxF x, x$ does not hold then $\lnot\Boxstar x$ is true.
\begin{Observation}\label{obs:tauto}
  Let $\phi\dfn\Psi \land \Phi$ be an $\LTL^{\Box_P,\Box_F,\supBoxstar}$ formula. Then any clause $C$ of $\Phi$ containing $\lnot\Boxstar x$ and (at least) one of $\Box_P x$, $\Box_F x$ or $x$ for some variable $x\in\Vars\phi$ is tautological and can thus be removed from $\phi$
  (without changing the satisfiability of $\phi$).
\end{Observation}

Observe that the tautological clauses above are exactly the clauses that are satisfied by every consistent assignment. It follows that once these clauses are removed from the formula, it holds that for every clause $C$ of $\phi$ there is a consistent assignment $\theta$ such that $C$ is not satisfied by $\theta$. This observation will be crucial for our detection algorithms described below.
\begin{theorem}\label{thm:detect-horn}
  For every $\OOO \subseteq \{\Boxstar, \Box_P, \Box_F\}$, $\Detect{\OOO}{\HORN}$ is in $\FPT$.
\end{theorem}
\begin{proof}
  Let $\OOO \subseteq \{\Boxstar, \Box_P, \Box_F\}$.
  We will reduce $\Detect{\OOO}{\HORN}$ to the problem $\VC$ which is well-known to be fixed-parameter tractable  (parameterised by the solution size) and which can actually be solved very efficiently in time $O(1.2738^km)$~\cite{ChenKX10}, where $k$ is the size of the vertex cover and $m$ the number of edges in the input graph. 
  Recall that given an undirected graph $G$ and an integer $k$, $\VC$ asks whether there is a subset $C \subseteq V(G)$ of size at most $k$ (which is called a vertex cover of $G$) such that $C \cap e \neq \emptyset$ for every $e \in E(G)$.
  Given an $\LTL^{\OOO}$ formula $\phi\dfn\Psi \land \Boxstar \Phi$, we will construct an undirected graph $G$ such that $\phi$ has a strong $\HORN$-backdoor of size at most $k$ if and only if $G$ has a vertex cover of size at most $k$. The graph $G$ has vertices $\Vars{\phi}$ and there is an edge between two vertices $x$ and $y$ in $G$ if and only if there is a clause that contains at least two literals from $\{x,y\} \cup \SB Ox, Oy \SM O \in \OOO \SE$. Note that if $x=y$, the graph $G$ contains a self-loop.
  We claim that a set $X \subseteq \Vars\phi$ is a strong $\HORN$-backdoor if and only if $X$ is a vertex cover of $G$.

  Towards showing the forward direction, let $X \subseteq \Vars{\phi}$ be a strong $\HORN$-backdoor set of $\phi$. 
  We claim that $X$ is also a vertex cover of $G$. Suppose for a contradiction that $X$ is not a vertex cover of $G$, i.e., there is an edge $\{x,y\} \in E(G)$ such that $X \cap \{x,y\}=\emptyset$. Because $\{x,y\} \in E(G)$, we obtain that there is a clause $C$ in $\Phi$ that contains at least two literals from $\{x,y\} \cup \SB Ox,Oy \SM O \in \OOO \SE$.
  Moreover, because of Observation~\ref{obs:tauto} there is a consistent assignment $\theta\colon X\cup\SB O x\mid x\in X \land O \in \OOO\}\to\{0,1\}$ that falsifies all literals of $C$ over variables in $X$. Consequently, $\phi[\theta]$ contains a sub-clause of $C$ that still contains at least two literals from $\{x,y\} \cup \SB Ox,Oy \SM O \in \OOO \SE$. Hence, $\phi[\theta] \notin \HORN$, contradicting our assumption that $X$ is a strong $\HORN$-backdoor set of $\phi$.
  
  Towards showing the reverse direction, let $X \subseteq V(G)$ be a vertex cover of $G$. We claim that $X$ is also a strong $\HORN$-backdoor of $\phi$. Suppose for a contradiction that this is not the case, then that there is an (consistent) assignment $\theta\colon X\cup\{O x\mid x\in X \land O \in \OOO\}\to\{0,1\}$ and a clause $C$ in $\phi[\theta]$ containing two positive literals say over variables $x$ and $y$. We obtain that $C$ contains at least two positive literals from $\{x,y\} \cup \SB Ox,Oy \SM O \in \OOO \SE$ and hence $G$ contains the edge $\{x,y\}$, contradicting our assumption that $X$ is a vertex cover of $G$.\qed
\end{proof}

\begin{theorem}\label{thm:detect-krom}
  For every $\OOO \subseteq \{\Boxstar, \Box_P, \Box_F\}$, $\Detect{\OOO}{\KROM}$ is in $\FPT$.
\end{theorem}
\begin{proof}
  Let $\OOO \subseteq \{\Boxstar, \Box_P, \Box_F\}$.
  We will reduce $\Detect{\OOO}{\KROM}$ to the $\threeHS$ problem, which is well-known to be fixed-parameter tractable (parameterised by the solution size)~\cite{AbuKhzam10}. Recall that given a universe $U$, a family $\FFF$ of subsets of $U$ of size at most three, and an integer $k$, $\threeHS$ asks whether there is a subset $S \subseteq U$ of size at most $k$ (which is called a hitting set of $\FFF$) such that $S \cap F \neq \emptyset$ for every $F \in \FFF$. Given an $\LTL^{\OOO}$ formula $\phi\dfn\Psi \land \Boxstar \Phi$, we will construct a family $\FFF$ of subsets (of size at most three) of a universe $U$ such that $\phi$ has a strong $\KROM$\hy backdoor of size at most $k$ if and only if $\FFF$ has a hitting set of size at most $k$. The universe $U$ is equal to $\Vars\phi$ and $\FFF$ contains the set $\Vars{C}$ for every set $C$ of exactly three literals contained in some clause of $\Phi$. 
  We claim that a set $X \subseteq \Vars\phi$ is a strong $\KROM$\hy backdoor if and only if $X$ is a hitting set of $\FFF$.

  Towards showing the forward direction, let $X \subseteq \Vars\phi$ be a strong $\KROM$\hy backdoor set of $\phi$ and suppose for a contradiction that there is a set $F \in \FFF$ such that $X \cap F=\emptyset$. It follows from the construction of $\FFF$ that $\Phi$ contains a clause $C$ containing at least three literals over the variables in $F$. Moreover, because of Observation~\ref{obs:tauto} there is a consistent assignment $\theta\colon X\cup\SB O x\mid x\in X \land O \in \OOO\}\to\{0,1\}$ that falsifies all literals of $C$ over variables in $X$. Consequently, $\phi[\theta]$ contains a sub-clause of $C$ that still contains at least three literals over the variables in $F$. Hence, $\phi[\theta]\notin \KROM$, contradicting our assumption that $X$ is a strong $\KROM$\hy backdoor set of $\phi$.
  
  Towards showing the reverse direction, let $X \subseteq U$ be a hitting set of $\FFF$ and suppose for contradiction that there is an (consistent) assignment $\theta\colon X\cup\{O x\mid x\in X\land O \in \OOO\}\to\{0,1\}$ and a clause $C$ in $\phi[\theta]$ containing at least three literals. Let $C'$ be a set of at exactly three literals from $C$. It follows from the construction of $\FFF$, that $\FFF$ contains the set $\Vars{C'}$, however, $\Vars{C'} \cap X=\emptyset$ contradicting our assumption that $X$ is a hitting set of $G$.\qed
\end{proof}

Now we have seen that the detection of both classes is very simple, i.e., in $\FPT$ time. Next we turn towards to the backdoor set evaluation problem and will first investigate the for the class $\HORN$ the problem lies in $\FPT$.

\section{Backdoor set evaluation}
\subsection{Formulas using only the always operator}

We showed in the previous section that strong backdoors can be found to the classes $\HORN$ and $\KROM$ in $\FPT$ time. This result holds independently of the considered temporal operators. In this section we will consider the question whether once a backdoor set is found it can also be used to efficiently decide the satisfiability of a formula in the case of formulas restricted to the $\Boxstar$ operator.
We will show that this is indeed the case for the class of $\HORN$ formulas but not for $\KROM$ formulas. Our tractability result for $\HORN$ formulas largely depends on the special semantics of formulas restricted to the $\Boxstar$ operators. Hence, we will start by showing some properties of these formulas necessary to obtain our tractability result.

\iflong
Let $\M=(\Z,<,V)$ be a temporal interpretation. We denote by $\Vars\M$ the set of propositions (in the following referred to by variables) for which $V$ is defined. For a set of variables $X \subseteq \Vars{\M}$, we denote by $\M_{|X}$ the \emph{projection} of $\M$ onto $X$, i.e., the temporal interpretation $\M_{|X}=(\Z,<,V_{|X})$, where $V_{|X}$ is only defined for the variables in $X$ and $V_{|X}(x)=V(x)$ for every $x \in X$.  
For an integer $z$, we denote by $\assign(\M,z)$ the assignment $\theta \colon \Vars\M \rightarrow\{0,1\}$ holding at world $z$ in $\M$, i.e., $\theta(v)=1$ if and only if $z \in \M(v)$ for every $v \in \Vars\M$. Moreover, for a set of worlds $Z \subseteq \Z$ we denote by $\assign(\M,Z)$ the set of all assignments ocurring in some world in $Z$ of $\M$, i.e., $\assign(\M,Z)\dfn\SB \assign(\M,z) \SM z \in Z \SE$. We also set $\assign(\M)$ to be $\assign(\M,\Z)$.
For an assignment $\theta \colon X \to \{0,1\}$, we denote by $\worlds(\M,\theta)$ the set of all worlds $z \in \Z$ of $\M$ such that $\assign(\M,z)$ is equal to $\theta$ on all variables in $X$. 
\fi

Let $\varphi\dfn\Psi \land \Boxstar \Phi \in \LTLstar_\CNF$.
We denote by $\propF(\Phi)$ the propositional CNF formula obtained from $\Phi$ after replacing each occurrence of $\Boxstar x$ in $\Phi$ with a fresh propositional variable (with the same name). 
For a set of variables $V$ and a set of assignments $\Alpha$ of the variables in $V$, we denote by $\glassign(\Alpha,V) \colon \SB \Boxstar v \SM v \in V \SE \to \{0,1\}$ the assignment defined by setting $\glassign(\Alpha,V)(\Boxstar v)=1$ if and only if $\alpha(v)=1$ for every $\alpha \in \Alpha$. Moreover, if $\theta \colon V \to \{0,1\}$ is an assignment of the variables in $V$, we denote by $\glassign(\Alpha,V,\theta)$ the assignment defined by setting $\glassign(\Alpha,V,\theta)(v)=\theta(v)$ and $\glassign(\Alpha,V,\theta)(\Boxstar v)=\glassign(\Alpha,V)(\Boxstar v)$ for every $v \in V$.
For a set $\Alpha$ of assignments over $V$ and an assignment $\theta \colon V' \to \{0,1\}$ with $V' \subseteq V$, we denote by $\Alpha(\theta)$ the set of all assignments $\alpha \in \Alpha$ such that $\alpha(v)=\theta(v)$ for every $v \in V'$, i.e., for all assignments from $\Alpha$ the values corresponding to $V'$ are overwritten with the one of $\theta'$.

For a set $\Alpha$ of assignments over some variables $V$ and a subset $V' \subseteq V$, we denote by $\Alpha|_{V'}$ the \emph{projection} of $\Alpha$ onto $V'$, i.e., the set of assignments $\alpha \in \Alpha$ restricted to the variables in $V'$.

Intuitively the next lemma describes the translation of a temporal model into separate satisfiability checks for propositional formulas.

\begin{lemma}\label{lem:boxstar-assign}
  Let $\varphi\dfn \Psi \land \Boxstar \Phi \in\LTLstar$. Then, $\varphi$ is satisfiable if and only if there is a set $\Alpha$ of assignments of the variables in $\varphi$ and an assignment $\alpha_0 \in \Alpha$ such that:
  $\alpha_0$ satisfies $\Psi$ and for every assignment $\alpha \in \Alpha$ it holds that $\glassign(\Alpha,\Vars\varphi,\alpha)$ satisfies the propositional formula $\propF(\Phi)$.
\end{lemma}

\iflong 
\begin{proof}
  Towards showing the forward direction assume that $\varphi\dfn\Psi\land \Boxstar \Phi$ is satisfiable and let $\M$ be a temporal interpretation witnessing this. We claim that the set of assignments $\Alpha\dfn\assign(\M)$ together with the assignment $\alpha_0\dfn\assign(\M,0)$ satisfy the conditions of the lemma.

  Towards showing the reverse direction assume that $\Alpha\dfn\{\alpha_0,\dotsc,\alpha_{|\Alpha|}\}$ is as given in the statement of the lemma.
  We claim that the temporal interpretation $\M$ defined below satisfies the formula $\varphi$. Let $\Z_{<0}$ be the set of all integers smaller than $0$ and let $\Z_{>|\Alpha|}$ be the set of all integers greater than $|\Alpha|$. Then for every variable $v \in \Vars\varphi$, the set $\M(v)$ contains the set $\SB z \SM\alpha_z(v)=1 \land 0 \leq z \leq |\Alpha|\SE$. Moreover, if $\alpha_0(v)=1$, $\M(v)$ also contains the set $\Z_{<0}$ and if $\alpha_{|\Alpha|}(v)=1$, $\M(v)$ additionally contains the set $\Z_{>|\Alpha|}$. It is easy to verify that $\M,0 \models \varphi$. \qed
\end{proof}
\fi

Informally, the following lemma shows that for deciding the satisfiability of an $\LTLstar$ formula,
we only need to consider sets of assignments $\Alpha$, whose size is linear (instead of exponential) in the number of variables.
\begin{lemma}\label{lem:boxstarls}
  Let $\varphi \dfn\Psi \land \Boxstar \Phi \in \LTLstar$ and $X \subseteq \Vars\varphi$.
  Then $\varphi$ is satisfiable if and only if there is a set $\Theta$ of assignments of the variables in $X$, an assignment $\theta_0 \in \Theta$, a set $\Alpha$ of assignments of the variables in $\Vars\varphi$, and an assignment $\alpha_0 \in \Alpha$ such that:
  \begin{itemize}
  \item[(C1)] the set $\Theta$ is equal to $\Alpha_{|X}$,
  \item[(C2)] the assignment $\theta_0$ is equal to $\alpha_0|_X$,
  \item[(C3)] $\Alpha$ and $\alpha_0$ satisfy the conditions stated in Lemma~\ref{lem:boxstar-assign}, and
  \item[(C4)] $|\Alpha(\theta)| \leq |\Vars\varphi\setminus X|+1$ for every $\theta \in \Theta$.
  \end{itemize}
\end{lemma}

\iflong
\begin{proof}
  Note that the reverse direction follows immediately from Lemma~\ref{lem:boxstar-assign}, because the existence of the set of assignments $\Alpha$ and the assignment $\alpha_0$ satisfying condition (C3) imply the satisfiability of $\varphi$.

  Towards showing the forward direction assume that $\varphi$ is satisfiable. Because of Lemma~\ref{lem:boxstar-assign} there is a set $\Alpha$ of assignments of the variables in $\varphi$ and an assignment $\alpha_0 \in \Alpha$ that satisfy the conditions of Lemma~\ref{lem:boxstar-assign}. Let $\Theta$ be equal to $\Alpha_{|X}$ and $\theta_0$ be equal to $\alpha_0|_X$. Observe that setting $\Theta$ and $\theta_0$ in this way already satisfies (C1) to (C3). We will show that there is a subset of $\Alpha$ that still satisfies (C1)--(C3) and additionally (C4).
  Towards showing this consider any subset $\Alpha'$ of $\Alpha$ that satisfies the following three conditions: (1) $\alpha_0 \in \Alpha'$, (2) for every $\theta \in \Theta$ it holds that $\Alpha'(\theta)\neq \emptyset$, and (3) for every variable $v$ of $\varphi$ and every $b \in \{0,1\}$ it holds that there is an assignment $\alpha \in \Alpha$ with $\alpha(v)=i$ if and only if there is an assignment $\alpha' \in \Alpha'$ with $\alpha'(v)=i$.
  Note that conditions (1) and (2) ensure that $\Alpha'$ satisfies (C1) and (C2) and condition (3) ensures (C3).
  Hence, any subset $\Alpha'$ satisfying conditions (1)--(3) still satisfies (C1)--(C3). It remains to show how to obtain such a subset $\Alpha'$ that additionally satisfies (C4).
  We define $\Alpha'$ as follows. Let $\Alpha_0'$ be a subset of $\Alpha$ containing $\alpha_0$ as well as one arbitrary assignment $\alpha \in \Alpha(\theta)$ for every $\theta \in \Theta$. 
  Note that $\Alpha_0'$ already satisfies conditions (1) and (2) as well as condition (3) for every variable $v \in X$. Observe furthermore that if there is a variable $v$ of $\varphi$ such that condition (3) is violated by $\Alpha_0'$ then it is sufficient to add at most one additional assignment to $\Alpha_0'$ in order to satisfy condition (3) for $v$. Let $\Alpha'$ be obtained from $\Alpha_0'$ by adding (at most $|\Vars\varphi\setminus X|$) assignments in order to ensure condition (3) for every variable $v \in \Vars\varphi\setminus X$. Then $\Alpha'$ satisfies the conditions of the lemma. \qed
\end{proof}
\fi

We are now ready to show tractability for the evaluation of strong $\HORN$\hy backdoor sets.
\begin{theorem}\label{thm:bd-eval-horn}
  $\BDeval{\supBoxstar}{\HORN}$ is in $\FPT$.
\end{theorem}
\begin{proof}
  Let $\varphi\dfn\Psi \land \Boxstar \Phi \in\LTLstar$ and let $X \subseteq \Vars\varphi$ be a strong $\HORN$\hy backdoor of $\varphi$.  The main idea of the algorithm is as follows:
  For every set $\Theta$ of assignments of the variables in $X$ and every $\theta_0 \in \Theta$, we will construct a propositional $\HORN$\hy formula $F_{\Theta,\theta_0}$, which is satisfiable if and only if there is a set $\Alpha$ of assignments of the variables in $\Vars\varphi$ and an assignment $\alpha_0 \in \Alpha$ satisfying the conditions of Lemma~\ref{lem:boxstarls}. 
  It then follows from Lemma~\ref{lem:boxstarls} that $\varphi$ is satisfiable if and only if there is such a set $\Theta$ of assignments and an assignment $\theta_0 \in \Theta$ for which $F_{\Theta,\theta_0}$ is satisfiable.
  Because there are at most $2^{2^{|X|}}$ such sets $\Theta$ and at most $2^{|X|}$ such assignments $\theta_0$ and for each of these sets the formula $F_{\Theta,\theta_0}$ is a $\HORN$\hy formula, it follows that checking whether there are $\Theta$ and $\theta_0$ such that the formula $F_{\Theta,\theta_0}$ is satisfied (and therefore decide the satisfiability of $\varphi$) can be done in time $O(2^{2^{|X|}}\cdot 2^{X}\cdot |F_{\Theta,\theta_0}|)$. Since we will show below that the length of the formula $F_{\Theta,\theta_0}$ can be bounded by an (exponential) function of $|X|$ times a polynomial in the input size, i.e., the formula $\varphi$, this implies that $\BDeval{\supBoxstar}{\HORN}$ is in $\FPT$. 

  The remainder of the proof is devoted to the construction of the formula $F_{\Theta,\theta_0}$ for a fixed set of assignments $\Theta$ and a fixed assignment $\theta_0 \in \Theta$ (and to show that it enforces the conditions of Lemma~\ref{lem:boxstarls}).

  Let $R\dfn\Vars\varphi\setminus X$ and $r\dfn|R|+1$. For a propositional formula $F$, a subset $V \subseteq \Vars{F}$, an integer $i$ and a label $s$, we denote by $\mcopy(F,V,i,s)$ the propositional formula obtained from $F$ after replacing each occurrence of a variable $v \in V$ with a novel variable $v^i_s$.
  We need the following auxiliary formulas.
  For every $\theta \in \Theta \setminus \theta_0$, let $F_{\Theta,\theta_0}^\theta$ be the formula:
  $$
  \bigwedge_{1\leq i \leq r}\mcopy(\propF(\Phi[\glassign(\Theta,X,\theta)]),R,i,\theta).
  $$
  Moreover, let $F_{\Theta,\theta_0}^{\theta_0}$ be the formula:
  \begin{align*}
    \mcopy(\Psi[\theta_0] \land\propF(\Phi[\glassign(\Theta,X,\theta_0)]),R,1,\theta_0) \land &\\
    \bigwedge_{2\leq i \leq r}\mcopy(\propF(\Phi[\glassign(\Theta,X,\theta_0)]),R,i,\theta_0).&
  \end{align*}

  Observe that because $X$ is a strong $\HORN$\hy backdoor set (and the formula $\Psi$ only consists of unit clauses), it holds that the formula $F_{\Theta,\theta_0}^\theta$ is $\HORN$ for every $\theta \in \Theta$.

  We also need the propositional formula $\Fcons$ that enforces the consistency between the propositional variables $\Boxstar x$ and the variables in $\SB x_\theta^i \SM \theta \in \Theta \land 1\leq i \leq r \SE$ for every $x \in \Vars\varphi \setminus X$. The formula $\Fcons$ consists of the following clauses: for every $\theta \in \Theta$, $i$ with $1 \leq i \leq r$, and $v \in R$, the clause $\Boxstar v \rightarrow v_\theta^i=\lnot \Boxstar v \lor v_\theta^i$ and for every $v \in R$ the clause
  $$\lnot \Boxstar v \rightarrow \bigvee_{\theta \in \Theta
    \land 1 \leq i \leq r}\lnot v_\theta^i=\Boxstar v \lor \bigvee_{\theta \in \Theta
    \land 1 \leq i \leq r}\lnot v_\theta^i.$$
Observe that $\Fcons$ is a $\HORN$ formula.

  Finally the formula $F_{\Theta,\theta_0}$ is defined as:
  $\bigwedge_{\theta \in \Theta}F_{\Theta,\theta_0}^\theta \land \Fcons.$
  
  \noindent Note that $F_{\Theta,\theta_0}$ is $\HORN$ and the length of
  $F_{\Theta,\theta_0}$ is at most 
  \begin{align*}
    &|F_{\Theta,\theta_0}| \leq  \sum_{\theta \in \Theta}|F_{\Theta,\theta_0}^{\theta}|+|\Fcons|\\
    & \leq |2^{|X|}(|\Vars\varphi \setminus X|+1)(|\Phi|+|\Psi|)+2\cdot 2^{|X|}\cdot(|\Vars\varphi\setminus X|+1)^2
  \end{align*}
  and consequently bounded by a function of $|X|$ times a polynomial in
  the input size. It is now relatively straightforward to verify that
  $F_{\Theta,\theta}$ is satisfiable if and only if there is a set
  $\Alpha$ of assignments of the variables in $\Vars\varphi$ and an
  assignment $\alpha_0 \in \Alpha$ satisfying the conditions of
  Lemma~\ref{lem:boxstarls}. Informally, 
  for every $\theta \in \Theta$,
  each of the $r$ copies of the formula
  $\propF(\Phi[\glassign(\Theta,X,\theta)])$ represents one of the at
  most $r$ assignments in $\Alpha(\theta)$,
  the formula $F_{\Theta,\theta_0}^{\theta_0}$ ensures (among other
  things) that the assignment choosen for $\alpha_0$ satisfies $\Psi$
  and the formula $\Fcons$ ensures that the ``global assignments''
  represented by the
  propositional variables $\Boxstar x$ is consistent with the set of local
  assignments in $\Alpha$ represented by the variables in
  $\SB x_\theta^i \SM \theta \in \Theta \land 1\leq i \leq r \SE$ for every $x \in \Vars\varphi \setminus X$. \qed
\end{proof}

Surprisingly the next result will show that $\KROM$ formulas turn out to be quite challenging. Backdoor set evaluation of this class of formulas is shown to be $\paraNP$-complete which witnesses an intractability degree in the parameterised sense.

\begin{theorem}\label{thm:eval-boxstar-krom}
  $\BDeval{\supBoxstar}{\KROM}$ is $\paraNP$-complete.
\end{theorem}
\begin{proof}
  The membership in $\paraNP$ follows because the satisfiability of $\LTL_\CNF^{\supBoxstar}$ can be decided in $\NP$~\cite[Table 1]{ArtaleKRZ13}.

  We show $\paraNP$-hardness of $\BDeval{\supBoxstar}{\KROM}$ by giving a polynomial time reduction from the $\NP$-hard problem $\threeCOL$ to $\BDeval{\supBoxstar}{\KROM}$ for backdoors of size two. In $\threeCOL$ one asks whether a given input graph $G=(V,E)$ has a colouring $f \colon V(G)\rightarrow \{1,2,3\}$ of its vertices with at most three colours such that $f(v)\neq f(u)$ for every edge $\{u,v\}$ of $G$. Given such a graph $G=(V,E)$, we will construct an $\LTLstar_\CNF$ formula $\phi\dfn\Psi \land \Boxstar\Phi$, which has a strong $(\KROM, \Boxstar)$-backdoor $B$ of size two, such that the graph $G$ has a $3$-colouring if and only if $\phi$ is satisfiable.

  For the remainder we will assume that there exists an arbitrary but fixed ordering of the vertices $V(G)=\{v_1,\dotsc,v_n\}$. Further for the construction we assume w.l.o.g.~that any undirected edge $e=\{v_i,v_j\}\in E$ follows this ordering, i.e., $i<j$. The formula $\phi$ contains the following variables:
  \begin{itemize}
  \item[(V1)] The variables $b_1$ and $b_2$. These variables make up the backdoor set $B$, i.e., $B\dfn\{b_1,b_2\}$.
  \item[(V2)] For every $i$ with $1 \leq i \leq n$, the variable $v_i$.
  \item[(V3)] For every $e=\{v_i,v_j\} \in E(G)$ with $1 \leq i,j \leq n$ the variables  $e_{v_iv_j}^{b_1b_2}$, $e_{v_iv_j}^{\bar{b}_1b_2}$, and $e_{v_iv_j}^{b_1\bar{b}_2}$.
  \end{itemize}
  
  We set $\Psi$ to be the empty formula and the formula $\Phi$ contains the following clauses:
  \begin{itemize}
  \item[(C1)] For every $i$ with $1 \leq i \leq n$, the clause $\lnot \Boxstar v_i$. Informally, this clause ensures that $v_i$ has to be false at least at one world, which will later be used to assign a color to the vertex $v_i$ of $G$. Observe that the clause is $\KROM$.
  \item[(C2)] For every $e=\{v_i,v_j\} \in E(G)$ with $1 \leq i,j \leq n$, the clauses $v_i \lor \Boxstar e_{v_iv_j}^{b_1b_2} \lor b_1 \lor b_2$, $v_i\lor \Boxstar e_{v_iv_j}^{\bar{b}_1b_2} \lor \lnot b_1 \lor b_2$, and $v_i \lor \Boxstar e_{v_iv_j}^{b_1\bar{b}_2} \lor b_1 \lor \lnot b_2$ as well as the clauses $v_j \lor \lnot \Boxstar e_{v_iv_j}^{b_1b_2} \lor b_1 \lor b_2$, $v_j \lor \lnot \Boxstar e_{v_iv_j}^{\bar{b}_1b_2} \lor \lnot b_1 \lor b_2$, and $v_j \lor \lnot \Boxstar e_{v_iv_j}^{b_1\bar{b}_2} \lor b_1 \lor \lnot    b_2$.
    Observe that all of these clauses are $\KROM$ after deleting the variables in $B$. 
  \item[(C3)] The clause $\lnot b_1 \lor \lnot b_2$. Informally, this clause excludes the color represented by setting $b_1$ and $b_2$ to true. Observe that the clause is $\KROM$.
  \end{itemize}

  It follows from the definition of $\phi$ that $\phi[\theta] \in \LTL_\KROM^{\supBoxstar}$ for every assignment $\theta$ of the variables in $B$. Hence, $B$ is a strong $(\KROM, \Boxstar)$\hy backdoor of size two of $\phi$ as required. Since moreover $\phi$ can be constructed in polynomial time, it only remains to show that $G$ has a $3$-Colouring if and only if $\phi$ is satisfiable.
  \begin{figure}[t]
    \begin{center}
      \begin{tikzpicture}[x=.5cm, y=1cm]
        \tikzstyle{ver}=[circle,inner sep=2pt,draw]
        \tikzstyle{every node}=[]
        \tikzstyle{every edge}=[line width=1pt,draw]
        
        \begin{scope}
          \node[ver, label=below:$1$] (v1) at (0,0) {$v_1$};
          \node[ver, label=above:$2$] (v2) at (1,1) {$v_2$};
          \node[ver, label=below:$3$] (v3) at (2,0) {$v_3$};
          
          \draw
          (v1) edge (v2)
          (v2) edge (v3)
          (v1) edge (v3)
          ;
        \end{scope}
      \end{tikzpicture}\quad
      \raisebox{1.1cm}{\begin{tabular}{c@{$\;\;$}cc@{$\;\;$}ccc@{$\;\;$}ccc@{$\;\;$}ccc@{$\;\;$}ccc}\toprule
         & $b_1$ & $b_2$ & $v_1$ & $v_2$ & $v_3$ & $e_{v_1v_2}^{b_1b_2}$ &
         $e_{v_1v_2}^{\bar{b}_1b_2}$ & $e_{v_1v_2}^{b_1\bar{b}_2}$ &
         $e_{v_1v_3}^{b_1b_2}$ &
         $e_{v_1v_3}^{\bar{b}_1b_2}$ & $e_{v_1v_3}^{b_1\bar{b}_2}$ & $e_{v_2v_3}^{b_1b_2}$ &
         $e_{v_2v_3}^{\bar{b}_1b_2}$ & $e_{v_2v_3}^{b_1\bar{b}_2}$  \\\midrule
        1 & 0 & 0 & 0 & 1 & 1 & 1 & 0 & - & 1 & - & 0 & - & 1 & 0 \\
        2 & 1 & 0 & 1 & 0 & 1 & 1 & 0 & - & 1 & - & 0 & - & 1 & 0 \\
        3 & 0 & 1 & 1 & 1 & 0 & 1 & 0 & - & 1 & - & 0 & - & 1 & 0\\\bottomrule
      \end{tabular}}
    \end{center}
    \caption{Left: A graph $G$ with vertices $v_1$, $v_2$, and $v_3$ together with a $3$-Colouring given by the numbers above and below respectively of every vertex. 
      Right: A temporal interpretation $\M$ that corresponds to the given $3$-Colouring of $G$ and satisfies $\M \models\phi$ given as a table. Each row of the table corresponds to a world 
      as indicated by the first column of the table. Each column represents the assignments of a variable as indicated in the first row. A ``-'' indicates that the assignment is not fixed, i.e., the assignment does not influence whether $\M \models\phi$.}\label{fig:box-krom-hardness}
  \end{figure}
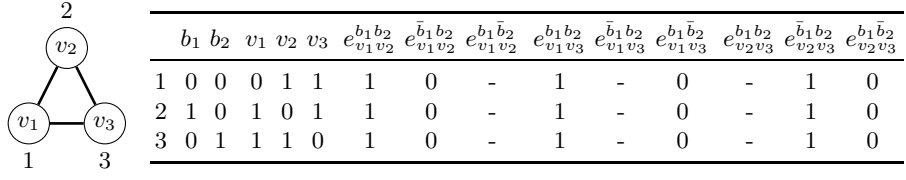  
\iflong    
  Towards showing the forward direction assume that $G$ has a $3$-Colouring and let $f \colon V(G) \rightarrow \{1,2,3\}$ be such a $3$-Colouring for $G$. We will show that $\phi$ is satisfiable by constructing a temporal interpretation $\M$ such that $\M \models\phi$. $\M$ is defined as follows:
  \begin{itemize}
  \item For every $i$ with $1 \leq i \leq n$, we set $\M(v_i)=\N\setminus f(v_i)$.
  \item We set $\M(b_1)=\{2\}$ and $\M(b_2)=\{3\}$.
  \item For every $e=\{v_i,v_j\} \in E(G)$:
    \begin{itemize}
    \item if $f(v_i)=1$ set $\M(e_{v_iv_j}^{b_1b_2})=\N$, else
      set $\M(e_{v_iv_j}^{b_1b_2})=\emptyset$.
    \item if $f(v_i)=2$ set $\M(e_{v_iv_j}^{\bar{b}_1b_2})=\N$, else
      set $\M(e_{v_iv_j}^{\bar{b}_1b_2})=\emptyset$.
    \item  if $f(v_i)=3$ set $\M(e_{v_iv_j}^{b_1\bar{b}_2})=\N$, else
      set $\M(e_{v_iv_j}^{b_1\bar{b}_2})=\emptyset$.
    \end{itemize}
  \end{itemize}
  
  An example for such a temporal interpretation resulting for a simple graph is illustrated in Figure~\ref{fig:box-krom-hardness}.
  Towards showing that $\M \models \phi$, we consider the different types of clauses given in (C1)--(C3).
  \begin{itemize}
  \item The clauses in (C1) hold because $\M, f(v_i) \not\models v_i$ for every $i$ with $1 \leq i \leq n$.
  \item For every $e=\{v_i,v_j\} \in E(G)$, we have to show that the clauses given in (C2) are satisfied for every world. 
    Because $f$ is a $3$-Colouring of $G$, we obtain that $f(v_i)\neq f(v_j)$. W.l.o.g. we assume in the following that $f(v_i)=1$ and $f(v_j)=2$. We first consider the clauses given in (C2) containing $v_i$. 
    Because $\M(v_i)=\N\setminus \{1\}$, it only remains to consider the world $1$.
    In this world $b_1$ and $b_2$ are false. It follows that all clauses containing either $\lnot b_1$ or $\lnot b_2$ are satisfied in this world.
    Hence, it only remains to consider clauses of the form $v_i \lor \Boxstar e_{v_iv_j}^{b_1b_2} \lor b_1 \lor b_2$. But these are satisfied because $f(v_i)=1$ implies that $\M(e_{v_iv_j}^{b_1b_2})=\N$.

    Consider now the clauses given in (C2) that contain $v_j$. Using the same argumentation as used above for $v_i$, we obtain that we only need to consider world $2$ and moreover we only need to consider clauses of the form $v_j \lor \lnot \Boxstar e_{v_iv_j}^{\bar{b}_1b_2}\lor \lnot b_1 \lor b_2$. Because $f(v_i)=1$, we obtain that $\M(e_{v_iv_j}^{\bar{b}_1b_2})=\emptyset$, which implies that also these clauses are satisfied.
  \item The clause $\lnot b_1 \lor \lnot b_2$ is trivially satisfied, because there is no world in which $b_1$ and $b_2$ holds simultaneously.
  \end{itemize}

  Towards showing the reverse direction assume that $\phi$ is satisfiable and let $\M$ be a temporal interpretation witnessing this. First note that because of the clauses added by C1, it holds that $\M(v_i)\neq \N$ for every $i$ with $1\leq i \leq n$. Let $w \colon V(G) \rightarrow \N$ be defined such that for every $i$ with $1 \leq i \leq n$, $w(v_i)$ is an arbitrary world in $\N\setminus \M(v_i)$. We define $f \colon V(G) \rightarrow\{1,2,3\}$ by setting:
  \begin{itemize}
  \item $f(v_i)=1$ if $\M,w(v_i) \not\models b_1 \lor b_2$,
  \item $f(v_i)=2$ if $\M,w(v_i) \not\models \lnot b_1 \lor b_2$, and
  \item $f(v_i)=3$ if $\M,w(v_i) \not\models b_1 \lor \lnot b_2$.
  \end{itemize}
  
  Note that because of the clause added by (C3), $f$ assigns exactly one color to every vertex $v_i$ of $G$. We claim that $f$ is a $3$-Colouring of $G$. To show this it suffices to show that for every $e=\{v_i,v_j\}\in E(G)$, it holds that $f(v_i)\neq f(v_j)$. Assume for a contradiction that this is not the case, i.e., there is an edge $e=\{v_i,v_j\}\in E(G)$ such that $f(v_i)=f(v_j)$. W.l.o.g.~assume furthermore that $f(v_i)=f(v_j)=1$. Consider the clause $v_i \lor\Boxstar e_{v_iv_j}^{b_1b_2} \lor b_1 \lor b_2$ (which was added by C2). Then, because of the definition of $w$ and $f$, we obtain that $\M, w(v_i) \not\models v_i \lor b_1 \lor b_2$. It follows that $\M,w(v_i) \models \Boxstar e_{v_iv_j}^{b_1b_2}$. Consider now the clause $v_j \lor \lnot \Boxstar e_{v_iv_j}^{b_1b_2} \lor b_1 \lor b_2$ (which was added by C2). Then, again because of the choice of $w$ and $f$, we obtain that $\M,w(v_j)\not\models v_j \lor b_1 \lor b_2$. Hence, $\M,w(v_j)\models \lnot \Boxstar e_{v_iv_j}^{b_1b_2}$ contradicting $\M,w(v_i)\models \Boxstar e_{v_iv_j}^{b_1b_2}$. This completes the proof of the theorem.
 \else
   The proof of the equivalence is in the technical report \cite{mso16}.
 \fi
\qed
\end{proof}

\subsection{Globally in the past and globally in the future}

Now we turn to a more flexible fragment where now we can talk about the past as well as about the future. Through this it is possible to encode $\NP$-complete problems into the $\HORN$-fragment yielding a $\paraNP$ lower bound of the problem.

\begin{theorem}\label{thm:eval-boxfboxp-horn}
  $\BDeval{\Box_F,\Box_P}{\HORN}$ is  $\paraNP$-complete.
\end{theorem}
\begin{proof}
  The membership in $\paraNP$ follows because the satisfiability of 
  $\LTL_\CNF^{\Box_F,\Box_P}$ can be decided in $\NP$~\cite[Table 1]{ArtaleKRZ13}.

  We show $\paraNP$-hardness of $\BDeval{\Box_F,\Box_P}{\HORN}$ by stating again a polynomial time reduction from $\threeCOL$ to $\BDeval{\Box_F,\Box_P}{\HORN}$ for backdoors of size four. In $\threeCOL$ one asks whether a given input graph $G=(V,E)$ has a colouring $f\colon V(G) \rightarrow\{1,2,3\}$ of its vertices with at most three colours such that $f(v)\neq f(u)$ for every edge $\{u,v\}$ of $G$. Given such a graph $G=(V,E)$, we will construct an $\LTL^{\Box_F,\Box_P}_\CNF$ formula $\phi\dfn\Psi \land \Boxstar\Phi$, which has a strong $(\HORN, \{\Box_F,\Box_P\})$-backdoor $B$ of size four, such that the graph $G$ has a $3$-colouring if and only if $\phi$ is satisfiable.

  For the remainder we will assume that $V(G)=\{v_1,\dotsc,v_n\}$ and $E(G)=\{e_1,\dotsc,e_m\}$. The formula $\phi$ contains the following variables:
  \begin{itemize}
  \item[(V1)] The variables $c_1,c_2,c_3,p_n'$ . These variables make up the backdoor set $B$, i.e., $B\dfn\{c_1,c_2,c_3,p_n'\}$.
  \item[(V2)] The variable $s$, which indicates the starting world. 
  \item[(V3)] For every $i$ with $1 \leq i \leq n$, three variables $v_i^1,v_i^2,v_i^3$.
  \item[(V4)] For every $i$ with $1 \leq i \leq n$ the variable $p_i$.
  \end{itemize}
  
  We set $\Psi$ to be the formula $s$ and the formula $\Phi$ contains the following clauses:
  \begin{itemize}
  \item[(C1)] The clauses $c_1 \lor c_2 \lor c_3$, $\lnot c_1 \lor\lnot c_2 \lor \lnot c_3$, $c_1 \lor \lnot c_2 \lor \lnot c_3$, $\lnot c_1 \lor \lnot c_2 \lor c_3$, and $\lnot c_1 \lor c_2 \lor\lnot c_3$. Informally, these clauses ensure that in every world it holds that exactly one of the variables $c_1,c_2,c_3$ is true. Note that $c_1 \lor c_2 \lor c_3$ is not $\HORN$, however, all of its variables are contained in the backdoor set $B$. 
  \item[(C2)] For every $i$ and $c$ with $1 \leq i \leq n$ and $1 \leq c \leq 3$, the clauses $v_i^c \rightarrow \Box_F v_i^c=\lnot v_i^c \lor \Box_F v_i^c$ and $v_i^c \rightarrow \Box_P v_i^c=\lnot v_i^c \lor \Box_P v_i^c$. Informally, these clauses ensure that the variable $v_i^c$ either holds in every world or in no world for every $i$ and $c$ as above. Observe that both of these clauses are $\HORN$.
  \item[(C3)] Informally, the following set of clauses ensures together that for every $i$ with $1\leq i \leq n$, it holds that $p_i$ is true in every world apart from the $i$-th world (where $p_i$ is false). Here, the first world is assumed to be the starting world.
    \begin{itemize}
    \item[(C3-1)] The clauses $s \rightarrow \lnot p_1=\lnot s \lor \lnot p_1$, $s \rightarrow\Box_F p_1=\lnot s \lor \Box_F p_1$, and $s \rightarrow \Box_P p_1=\lnot s \lor \Box_P p_1$. Informally, these clauses ensure that $p_1$ is only false in the starting world (and otherwise true).
    \item[(C3-2)] The clause $p_i \land \Box_F p_i\rightarrow \Box_F p_{i+1}=\lnot p_i \lor \lnot \Box_F p_i \lor \Box_F p_{i+1}$ for every $i$ with $1 \leq i < n$. Informally, these clauses (together with the clauses from C3-1) ensure that for every $i$ with $2 \leq i\leq n$, it holds that $p_i$ is true in every world after the $i$-th world.
    \item[(C3-3)] The clause $\lnot p_i \rightarrow \lnot\Box_F p_{i+1}=p_i \lor \lnot \Box_F p_{i+1}$ for every $i$ with $1 \leq i < n$. Informally, these clauses (together with the clauses from C3-1 and C3-2) ensure that for every $i$ with $2 \leq i\leq n$, it holds that $p_i$ is false at the $i$-th world. Observe that the clauses from C3-1 to C3-3 already ensure that $\lnot p_i \land\Box_F p_i$ holds if and only if we are at the $i$-th world of the model for every $i$ with $1 \leq i \leq n$.
    \item[(C3-4)] The clauses $\lnot p_n\land \Box_F p_n\rightarrow p_n'=p_n\lor \lnot \Box_F p_n \lor p_n'$ and $\lnot p_n\land \Box_F p_n \leftarrow p_n'=\lnot p_n\land \Box_F p_n \lor \lnot p_n'=(\lnot p_n \lor \lnot p_n') \land (\Box_F p_n \lor \lnot p_n')$. Informally, these clauses (together with the clauses from C3-1 to C3-3) ensure that $p_n'$ only holds in the $n$-th world of the model. Observe that all these clauses are $\HORN$ after removing the backdoor set variable $p_n'$.
    \item[(C3-5)] The clause $p_n' \rightarrow \Box_P p_n=\lnot p_n' \lor \Box_P p_n$. 
      Informally, this clause (together with the clauses from C3-1 to C3-4) ensures that $p_n$ is only false in the $n$-th world of the model.
    \item[(C3-6)] The clause $p_i \land \Box_P p_i \rightarrow \Box_P p_{i-1}=\lnot p_i \lor \lnot \Box_P p_i \lor \Box_P p_{i-1}$ for every $i$ with $2 \leq i \leq n$. Informally, these clauses (together with the clauses from C3-1 to C3-5) ensure that $p_i$ is true before the $i$-th world for every $i$ with $2 \leq i < n$.
     \end{itemize}
     Observe that all of the above clauses are $\HORN$ or become $\HORN$ after removing all variables from $B$. Note furthermore that all the above clauses ensure that $\Box_P p_i \land \Box_F p_i$ holds if and only if we are at the $i$-th world of the model for every $i$ with $1 \leq i \leq n$. 
  \item[(C4)] For every $i$ and $j$ with $1 \leq i \leq n$ and $1\leq j \leq 3$ the clauses $\Box_Fp_i \land \Box_P p_i \land v_i^j \rightarrow c_j=\lnot \Box_F p_i \lor\lnot \Box_P p_i \lor \lnot v_i^j \lor c_j$ and $\Box_F p_i \land \Box_P p_i \land c_j \rightarrow v_i^j=\lnot\Box_F p_i\lor \lnot \Box_P p_i \lor \lnot c_j \lor v_i^j$. Informally, these clauses ensure that in the $i$-th world for every $1 \leq i\leq n$, the variables $c_1$, $c_2$, $c_3$ are a copy of the variables $v_i^1$, $v_i^2$, $v_i^3$. Observe that all of these clauses are $\HORN$. 
  \item[(C5)] For every edge $e=\{v_i,v_j\} \in E(G)$ and every $c$ with $1 \leq c \leq 3$, the clause $\lnot v_i^c \lor \lnot v_j^c$. Informally, these clauses ensure that the $3$-partition (of the vertices of $G$) given by the (global) values of the variables $v_1^1,v_1^2,v_1^3,\dotsc,v_n^1,v_n^2,v_n^3$ is a valid $3$-Colouring for $G$. Observe that all of these clauses are $\HORN$. 
 \end{itemize}
  
  It follows from the definition of $\phi$ that $\phi[\theta] \in\LTL_\HORN^{\Box_F,\Box_P}$ for every assignment $\theta$ of the variables in $B$. Hence, $B$ is a strong $(\HORN,\{\Box_P,\Box_F\})$-backdoor of size four of $\phi$ as required.
  Since moreover $\phi$ can be constructed in polynomial time, it only remains to show that $G$ has a $3$-Colouring if and only if $\phi$ is satisfiable.
\iflong
  \begin{figure}
    \begin{center}
      \begin{tikzpicture}[xscale=1,yscale=1, node distance=1.8cm]
        \tikzstyle{ver}=[circle,inner sep=2pt,draw]
        \tikzstyle{every node}=[]
        \tikzstyle{every edge}=[line width=1pt,draw]
        
        \begin{scope}
          \draw 
          node[ver, label=below:$1$] (v1) {$v_1$}
          node[ver,above right of=v1, label=above:$2$] (v2) {$v_2$}
          node[ver,below right of=v2, label=below:$3$] (v3) {$v_3$}
          ;
          
          \draw
          (v1) edge (v2)
          (v2) edge (v3)
          (v1) edge (v3)
          ;
        \end{scope}
      \end{tikzpicture}
      \qquad
      \raisebox{1.25cm}{\setlength{\tabcolsep}{.1mm}\begin{tabular}{c@{$\;\;$}ccccc@{$\;\;$}ccc@{$\;\;$}ccc@{$\;\;$}ccc@{$\;\;$}ccc}\toprule
         & $s$ & $c_1$ & $c_2$ & $c_3$ & $p_n'$ & $v_1^1$ & $v_1^2$ & $v_1^3$ &
        $v_2^1$ & $v_2^2$ & $v_2^3$ & $v_3^1$ & $v_3^2$ & $v_3^3$ & $p_1$
        & $p_2$ & $p_3$ \\\midrule
        $<1$ & 0 & - & - & - & 0 & 1 & 0 & 0 & 0 & 1 & 0 &
        0 & 0 & 1 & 1 & 1 & 1 \\
        
        1 & 1 & 1 & 0 & 0 & 0 & 1 & 0 & 0 & 0 & 1 & 0 &
        0 & 0 & 1 & 0 & 1 & 1 \\
        2 & 0 & 0 & 1 & 0 & 0 & 1 & 0 & 0 & 0 & 1 & 0 &
        0 & 0 & 1 & 1 & 0 & 1 \\
        3 & 0 & 0 & 0 & 1 & 1 & 1 & 0 & 0 & 0 & 1 & 0 &
        0 & 0 & 1 & 1 & 1 & 0 \\
        
        $>3$ & 0 & - & - & - & 0 & 1 & 0 & 0 & 0 & 1 & 0 &
        0 & 0 & 1 & 1 & 1 & 1 \\\bottomrule
      \end{tabular}}
    \end{center}
    \caption{Left: A graph $G$ with vertices $v_1$, $v_2$, and $v_3$ together with a $3$-Colouring given by the numbers above and below respectively of every vertex. 
      Right: A temporal interpretation $\M$ that corresponds to the given $3$-Colouring of $G$ and satisfies $\M \models\phi$ given as a table. Each row of the table corresponds to a world (or a set of worlds) as indicated by the first column of the table. Each column represents the assignments of a variable as indicated in the first row. A ``-'' indicates that the assignment is not fixed, i.e., the assignment does not influence whether $\M \models\phi$.}
    \label{fig:box-horn-hardness}
  \end{figure}  
    
  Towards showing the forward direction assume that $G$ has a $3$-Colouring and let $f \colon V(G) \rightarrow \{1,2,3\}$ be such a $3$-Colouring for $G$. We will show that $\phi$ is satisfiable by constructing a temporal interpretation $\M$ such that $\M \models\phi$. $\M$ is defined as follows:
  \begin{itemize}
  \item For every $j$ with $1\leq j \leq 3$, we set $\M(c_j)=\SB i \SM f(v_i)=j\SE$.
  \item We set $\M(p_n')=\{n\}$.
  \item For every $i$ and $c$ with $1 \leq i \leq n$ and $1\leq c \leq3$, we set $\M(v_i^{c})=\Z$ if $c=f(v_i)$ and otherwise we set $\M(v_i^c)=\emptyset$.
  \item For every $i$ with $1 \leq i \leq n$, we set $\M(p_i)=\Z \setminus \{i\}$.
  \end{itemize}
  
  An example for such a temporal interpretation resulting for a simple graph is illustrated in Figure~\ref{fig:box-horn-hardness}.
  It is straightforward (but a little tedious) to verify that $\M \models \phi$ by considering all the clauses of $\phi$.

  Towards showing the reverse direction assume that $\phi$ is satisfiable and let $\M$ be a temporal interpretation witnessing this. We will start by showing the following series of claims for $\M$.
  \begin{itemize}
  \item[(M1)] For every $a \in \N$ exactly one of $\M, a \models c_1$, $\M, a \models c_2$, and $\M, a \models c_3$ holds.
  \item[(M2)] For every $i$, $c$, $a$, and $a'$ with $1 \leq i \leq n$, $1 \leq c \leq 3$, and $a,a' \in \N$, it holds that $\M, a \models v_i^c$ if and only if $\M, a' \models v_i^c$.
  \item[(M3)] For every $i$ with $1 \leq i \leq n$ and every $a \in \N$, it holds that $\M, a \models p_i$ if and only if $a \neq i$.
  \item[(M4)] For every $i$ and $j$ with $1 \leq i \leq n$ and $1 \leq j \leq 3$, it holds that $\M, i \models c_j$ if and only if $\M,i \models v_i^j$.
  \end{itemize}
  
  (M1) holds because of the clauses added by (C1). Towards showing (M2) consider the clauses added by (C2) and assume for a contradiction that there are $i$, $c$, $a$, and $a'$ as in the statement of (M2) such that w.l.o.g. $\M, a \models v_i^c$ but $\M, a' \not\models v_i^c$. Then, $a \neq a'$. If $a < a'$, then we obtain a contradiction because of the clause $v_i^c \rightarrow\Box_F v_i^c$ and if on the other hand $a' < a$, we obtain a contradiction to the clause $v_i^c \rightarrow \Box_P v_i^c$. This completes the proof of (M2).
  We will show (M3) with the help of the following series of claims.
  \begin{itemize}
  \item[(M3-1)] For every $a \in \N$ it holds that $\M,a \models p_1$ if and only if $a \neq 1$ (here we assume that $1$ is the starting world).
  \item[(M3-2)] For every $i$ and $a$ with $1 \leq i \leq n$, $a \in \N$, and $a>i$, it holds that $\M,a \models p_i$.
  \item[(M3-3)] For every $i$ with $1 \leq i \leq n$, it holds that $\M,i \not\models p_i$.
  \item[(M3-4)] For every $a \in \N$, it holds that $\M,a \models p_n'$ if and only if $a=n$.
  \item[(M3-5)] For every $a \in \N$, it holds that $\M,a \not\models p_n$ if and only if $a=n$.
  \end{itemize}

  Because of the clause $s\rightarrow \lnot p_1$ (added by C3-1) and the fact that $s \in \Psi$, we obtain that $\M,1 \not\models p_1$. Moreover, because of the clauses $s \rightarrow \Box_F p_1$ and $s \rightarrow \Box_P p_1$, we obtain that $\M,a \models p_1$ for every $a \neq 1$. This completes the proof for (M3-1). 

  We show (M3-2) via induction on $i$. The claim clearly holds for $i=1$ because of (M3-1). Now assume that the claim holds for $p_{i-1}$ and we want to show it for $p_i$. Because of the induction hypothesis, we obtain that $\M,i \models p_{i-1} \land\Box_Fp_{i-1}$. Moreover, because $\phi$ contains the clause $p_{i-1} \land \Box_F p_{i-1} \rightarrow \Box_F p_i$ (which was added by (C3-2)), we obtain that $\M,i \models \Box_F p_i$. This completes the proof of (M3-2). 

  We show (M3-3) via induction on $i$. The claim clearly holds for $i=1$ because of (M3-1). Now assume that the claim holds for $p_{i-1}$ and we want to show it for $p_i$. Because of  the induction hypothesis, we obtain that $\M,(i-1) \not\models p_{i-1}$. Furthermore, because of (M3-2), we know that $\M,i \models \Box_Fp_{i}$. Since $\phi$ contains the clause $\lnot p_{i-1} \rightarrow \lnot \Box_F p_{i}$ (which was added by (C3-3)), we obtain $\M, (i-1) \models \lnot \Box_F p_i$, which because $\M,i \models \Box_Fp_{i}$ can only hold if $\M,i \not\models p_i$. This completes the proof of (M3-3).

  Towards showing (M3-4), first note that because of (M3-2) and (M3-3), we have that $\M,a \models \lnot p_n \land \Box_F p_n$ if and only if $a=n$. Then, because of the clauses (added by C3-4) ensuring that $\lnot p_n \land \Box_F p_n \leftrightarrow p_n'$, the same applies to $p_n'$ (instead of $\lnot p_n \land \Box_F p_n$). This completes the proof of (M3-4).

  It follows from (M3-2) and (M3-3) that (M3-5) holds for every $a \in \N$ with $a \geq n$. Moreover, because of (M3-4), we have that $\M, n \models p_i'$. Because of the clause $p_n' \rightarrow \Box_P p_n$ (which was added by (C3-5)), we obtain $\M,a \models p_n$ for every $a<n$. This completes the proof of (M3-5).

  We are now ready to proof (M3). It follows from (M3-2) and (M3-3) that (M3) holds for every $i$ and $a$ with $a \geq i$. Furthermore, we obtain from (M3-5) that (M3) already holds if $i=n$. We complete the proof of (M3) via an induction on $i$ starting from $i=n$. Because of the induction hypothesis, we obtain that $\M,i+1 \models p_{i+1} \land \Box_P p_{i+1}$. Hence, because of the clause $p_{i+1}\land \Box_P p_{i+1}\rightarrow \Box_P p_{i}$ (added by (C3-6)), we obtain that $\M,i+1 \models \Box_P p_i$, which completes the proof of (M3).

  Towards showing (M4) first note that it follows from (M3) that $\M, i \models \Box_F p_i \land \Box_P p_i$.
  Now suppose that there are $i$ and $j$ such that either $\M, i \models c_j$ but $\M,i \not \models v_i^j$ or $\M, i \not\models c_j$ but $\M,i \models v_i^j$. In the former case, consider the clause $\Box_F p_i \land \Box_P p_i \land c_j \rightarrow v_i^j$ (which was added by (C4)). Since $\M, i \models \Box_F p_i \land \Box_P p_i$, we obtain that $\M,i \models v_i^j$; a contradiction. In the later case, consider the clause $\Box_F p_i \land \Box_P p_i \land v_i^j \rightarrow c_j$ (which was added by (C4)). Since $\M, i \models \Box_F p_i \land \Box_P p_i$, we obtain that $\M,i \models c_j$; again a contradiction. This completes the proof of the claims (M1)--(M4).

  It follows from (M1) and (M4) that for every $i$ and $a$ with $1 \leq i\leq n$ and $a \in \N$ there is exactly one $c$ with $1 \leq c \leq 3$, such that $\M,a \models v_i^c$. Moreover, because of (M2) the choice of $c$ is independent of $a$. Hence, the colouring $f$ that assigns the unique color $c$ to every vertex $v_i$ such that $\M,a \models v_i^c$ forms a partition of the vertex set of $G$. We claim that $f$ is also a valid $3$-Colouring of $G$. Assume not, then there is an edge $\{v_i,v_j\} \in E(G)$ such that $c=f(v_i)=f(v_j)$.
  Consider the clause $\lnot v_i^{c} \lor \lnot v_j^{c}$ (which was added by C5). Because of the definition of $f$, we obtain that $\M,a \not\models \lnot v_i^{c} \lor \lnot v_j^{c}$ for every $a \in \N$, a contradiction to our assumption that $\M \models \phi$.
 \else
  The proof can be found in the technical report \cite{mso16}.
 \fi
\qed
\end{proof}

\begin{corollary}\label{cor:eval-boxf-paraNP}
  Let $O\in\{\BoxF,\BoxP\}$ then $\BDeval{O}{\KROM}$ is $\paraNP$-complete.
\end{corollary}
\begin{proof}
  Follows from $\NP$-hardness of satisfiability of $\LTL^{O}_\KROM$ formulas \cite[Theorem~6]{ArtaleKRZ13}. \qed
\end{proof}

\section{Conclusion}
\begin{table}[t]
	\centering
	\begin{tabular}{lccc}\toprule
		Problem & Operator & $\HORN$ & $\KROM$ \\\midrule
		Detection&any&$\FPT$ (Thm.~\ref{thm:detect-horn}) & $\FPT$ (Thm.~\ref{thm:detect-krom}) \\
		Evaluation& $\Boxstar$ & $\FPT$ (Thm.~\ref{thm:bd-eval-horn}) & $\paraNP$-c.\ (Thm.~\ref{thm:eval-boxstar-krom}) \\
		&$\BoxF,\BoxP$ & $\paraNP$-c.\ (Thm.~\ref{thm:eval-boxfboxp-horn})& $\paraNP$-c.\ (above)\\
		&$\BoxF$ or $\BoxP$ &open & $\paraNP$-c.\ (Cor.~\ref{cor:eval-boxf-paraNP})
		\\\bottomrule
	\end{tabular}
	\caption{Results overview.}
\end{table}

We lift the well-known concept of backdoor sets from propositional logic up to the linear temporal logic $\LTL$. From the investigated cases we exhibit a parameterised complexity dichotomy for the problem of backdoor set evaluation. The evaluation into $\KROM$ formulas becomes in all cases $\paraNP$-complete and thus is unlikely to be solvable in $\FPT$ whereas the case of backdoor evaluation into the fragment $\HORN$ behaves different.
Allowing only $\Boxstar$ makes the problem fixed parameter tractable however allowing both $\BoxF$ and $\BoxP$ makes it $\paraNP$-complete. The last open case, i.e., the restriction to either $\BoxF$ or $\BoxP$ is open for further research and might yield an $\FPT$ result.

As a further research topic a solid definition of renamable horn and also qhorn formulas is of great interest. Furthermore the study of other operators beyond the investigated ones is open.

\bibliographystyle{plain}
\bibliography{bibliography}


\end{document}